\documentclass[11pt,letterpaper]{amsart}
\usepackage[margin=1in]{geometry}
\usepackage{color,ifpdf,latexsym,graphicx,hyperref,cite,enumerate}
\usepackage{todonotes}
\usepackage{cleveref}

\title{Existence and hardness of conveyor belts}
\author[Baird, Billey, Demaine, Demaine,  Eppstein,  Fekete, 
Gordon, Griffin,  Mitchell, Swanson]{%
  Molly Baird
\and
  Sara C. Billey
\and
  Erik D. Demaine
\and
  Martin L. Demaine
\and
  David Eppstein
\and
  S\'andor Fekete
\and
  Graham Gordon
\and
  Sean Griffin
\and
  Joseph S. B. Mitchell
\and
  Joshua P. Swanson
}

\thanks{This work was partially supported by the University of
Washington Graduate School and grants from the National Science
Foundation: DMS-1764012, CCF-1618301, and CCF-1616248.}

\newtheorem{theorem}{Theorem}
\newtheorem{lemma}[theorem]{Lemma}

\theoremstyle{definition}
\newtheorem{definition}[theorem]{Definition}
\newtheorem{example}[theorem]{Example}
\newtheorem{remark}[theorem]{Remark}

\let\epsilon=\varepsilon

\begin{document}

\begin{abstract}
An open problem of Manuel Abellanas asks whether every set of disjoint closed
unit disks in the plane can be connected by a conveyor belt, which
means a tight simple closed curve that touches the boundary of each
disk, possibly multiple times.
We prove three main results:
\begin{enumerate}
\item For unit disks whose centers are both $x$-monotone and $y$-monotone, or whose centers have $x$-coordinates that differ by at least two units, 
a conveyor belt always exists and can be found efficiently.
\item It is NP-complete to determine whether disks of varying radii have a conveyor belt,
and it remains NP-complete when we constrain the belt to touch disks exactly once.
\item Any disjoint set of $n$ disks of arbitrary radii can be augmented by $O(n)$ ``guide'' disks so that the augmented system has a conveyor belt touching each disk exactly once, answering a conjecture of Demaine, Demaine, and Palop.
\end{enumerate}
\end{abstract}

\maketitle

\section{Introduction}

In 2001 (later published in \cite{Abe-GRSME-08,Abe-EGC-11}), Manuel
Abellanas asked whether every finite collection of disjoint closed
disks in the plane can be spanned by a \emph{conveyor belt}. A
conveyor belt for such a collection of disks is a continuously
differentiable simple closed curve that touches the boundary of each disk at least
once, is disjoint from the disk interiors, and consists of arcs of the disks and bitangents between them. We may imagine such a curve as made out of an elastic band wrapped tightly around the disks. Consequently, all the disks and the conveyor belt with them can turn without slipping. Two disks have the same \emph{orientation} if they turn in the same direction when the conveyor belt is pulled, or equivalently if they are both on the same side of the belt. See \Cref{fig:large.belt}.

\begin{figure}[ht]
\centerline{\includegraphics[width=0.5\textwidth]{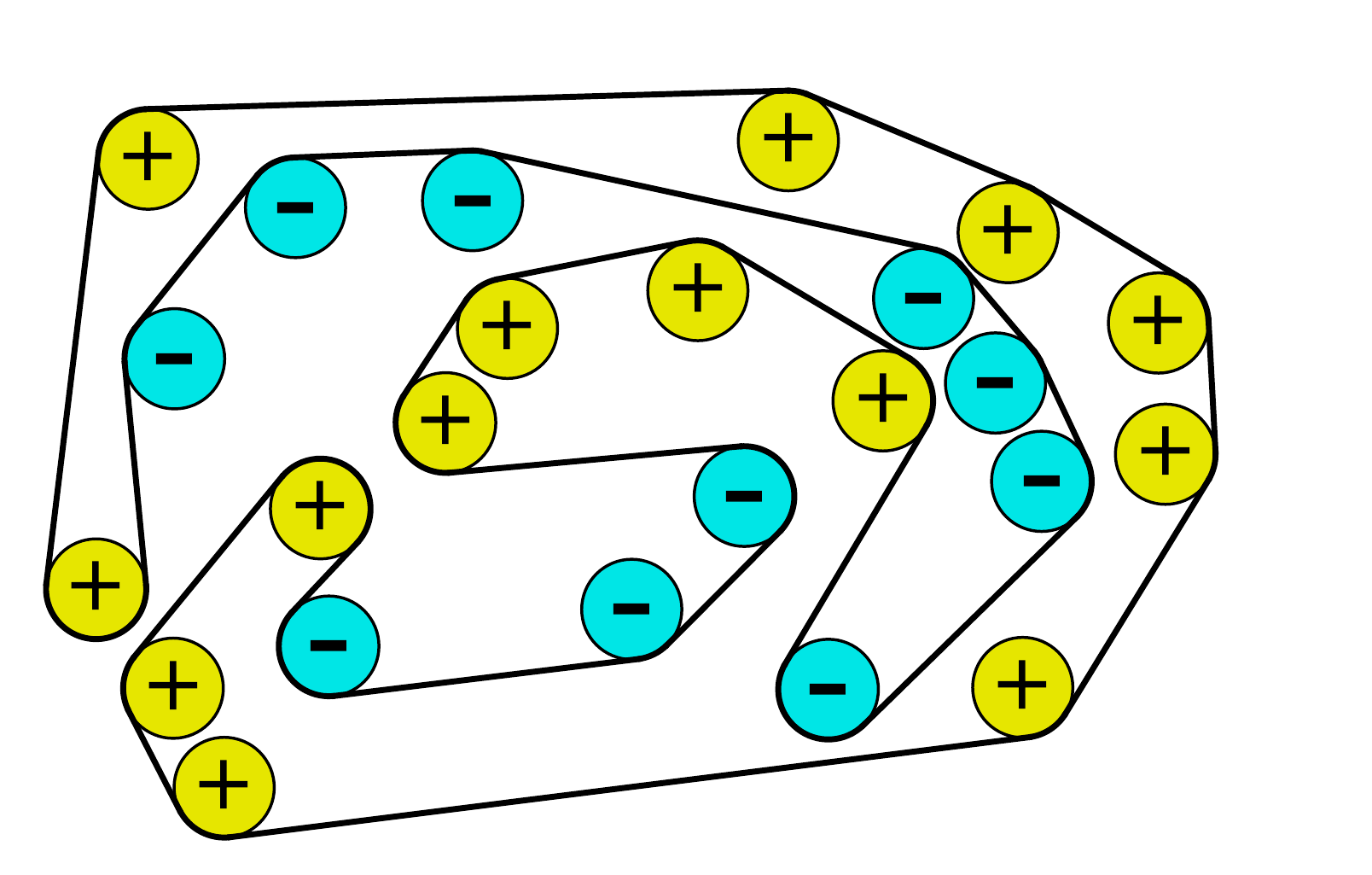}}
\caption{A conveyor belt on 24 nonoverlapping unit disks.  The colors and markings of the disks indicate their orientations.}
\label{fig:large.belt} 
\end{figure}

We review some of the history of Abellanas' question. While Abellanas' question remains open for unit disks, Tejel and Garc{\'\i}a found an example of non-unit disks that have no conveyor belt~\cite{DemDemPal-FE-10}. O'Rourke~\cite{ORo-EGC-11} relaxed the problem by allowing the curve to cross itself or wrap around some arcs of  disks more than once, with specified sets of disks of the same orientation. For his variant of the problem, not every system of disjoint unit disks has a conveyor belt, but for a belt of this type to exist it is sufficient for a certain \emph{hull-visibility graph} of the disks to be connected or for the disks to remain disjoint when expanded by a sufficiently large factor. Demaine, Demaine, and Palop~\cite{DemDemPal-FE-10,DemDem-TCS-15} designed puzzle fonts, specified by a system of disks per character, with a unique conveyor belt in the shape of that character. If the conveyor belt is not shown, decoding the font becomes a puzzle for the viewer.

In this paper, we prove several related results on conveyor belts:
\begin{enumerate}
\item For unit disks whose sorted orders by $x$- and by $y$-coordinates of their centers are the same (i.e., $xy$-monotone), there always exists a conveyor belt (\Cref{thm:separated-has-belt}), and a solution belt can be constructed in linear time after sorting (\Cref{thm:separated-linear-runtime}). The same method also applies to unit disks whose $x$-coordinates differ by two or more units and, more generally, to \emph{monotonically separated} configurations (\Cref{def:monotonically-separated}).
\item Strengthening the known result that not every system of non-unit disks has a conveyor belt~\cite{DemDemPal-FE-10}, we show that the decision problem is NP-complete for non-unit disks (\Cref{thm:multi-touch-completeness}).
\item A variation of the conveyor belt problem in which we constrain the belt to touch each disk exactly once (and disks may still have non-unit radii) is also NP-complete (\Cref{thm:one-touch-completeness}).
\item Both versions can be made to have a conveyor belt solution by adding $O(n)$ (non-unit) ``guide'' disks, answering a conjecture of Demaine, Demaine, and Palop~\cite{DemDemPal-FE-10} (\Cref{thm:guide-disks-suffice}). Conversely, $\Omega(n)$ guide disks are sometimes necessary (\Cref{thm:guide-disks-necessary}).
\end{enumerate}

\section{Preliminaries}
\label{sec:preliminaries}

Between any two disks $D_1$ and $D_2$, there are four \emph{bitangent} line segments, as depicted in \Cref{fig:4tangents}. In the case when the two disks $D_1$ and $D_2$ have centers $(x_1,y_1)$ and $(x_2,y_2)$ with $x_1<x_2$, define the \emph{lower bitangent} to be the bitangent which stays entirely below the line through the centers of $D_1$ and $D_2$. The \emph{upper bitangent} is defined similarly. Define the \emph{inner bitangents} as the two bitangents that cross the line between the centers of $D_1$ and $D_2$.

\begin{figure}[ht]
\centerline{\includegraphics[width=0.6\textwidth]{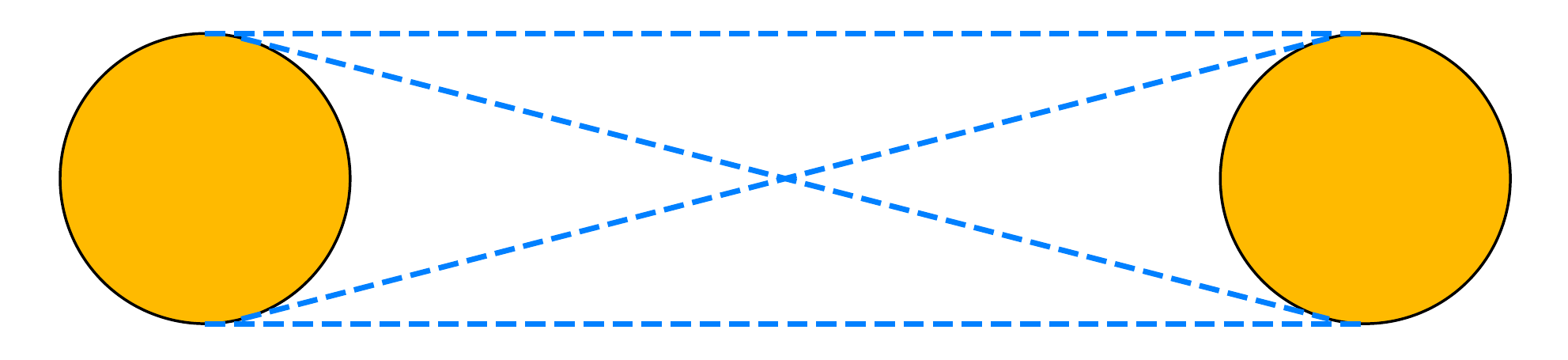}}
\caption{The four possible bitangents between two disks.}
\label{fig:4tangents}
\end{figure}

A conveyor belt for a collection of disks can be specified as follows. Between any two disks, the belt may travel along one of four bitangent line segments. The two disks have the same orientation if the bitangent does not pass through the line connecting their centers and they have different orientations otherwise. We call a bitangent \emph{unblocked} if it does not intersect any disk, except at its endpoints. A conveyor belt may be completely specified by the cyclic order in which it visits each disk, together with the subset of disks lying inside the conveyor belt.

As a warm-up, the following family of disk configurations have an easy-to-construct conveyor belt. In this case, the disks in fact have a conveyor belt that touches each of them exactly once, along an arc of nonzero length.

\begin{example}  
  Suppose the disks are in ``general position'' in the sense that no line intersects three disks. Let $P$ be any simple polygon with the disk centers as vertices; one often refers to such a $P$ as a (simple) polygonalization of the center points. Such a polygon $P$ must exist; a traveling salesman tour of the center points is one such (simple) polygonalization. Another approach is to choose one of the vertices of the convex hull of the circle centers and connect the remaining circle centers into a path in radial sorted order around the chosen vertex. To determine the orientations of the vertices, travel along $P$ cyclically and declare all disks where we turn left to be of one orientation and all disks where we turn right to be of the other orientation.
\end{example}

We will also need to review the Koebe--Andreev--Thurston circle packing theorem.

\begin{theorem}{\cite{Ziegler-1995-Steinitz}}
  For every planar graph $G$ there exists a system of interior-disjoint disks in the plane, corresponding one-to-one with the vertices of $G$, such that two vertices are adjacent in $G$ if and only if the corresponding two disks are tangent.
\end{theorem}
The centers and radii of these disks are algebraic numbers, but they can be of arbitrarily large degree, so it may not be straightforward to represent them exactly as objects in a symbolic algebra system~\cite{BanDevEpp-JGAA-15}. Nevertheless, there exist algorithms that can compute their coordinates numerically, to arbitrary precision, in time polynomial in the number of vertices of $G$ and in the number of bits of precision desired~\cite{M93,CS03}.

A \emph{maximal planar graph} is a graph embedded in the plane so that all of its faces are triangles, including the unbounded outer face. For maximal planar graphs, the corresponding circle packing is unique up to M\"obius transformations of the plane, which preserve circles and their tangencies. In this case, the circle packing can be chosen such that the three vertices of the outer face of the graph are three mutually tangent unit disks, with the rest of the disks all fitting into the triangular region bounded by these three unit disks. For this packing, the disks cannot vary more than exponentially in size: there exists a constant $\epsilon>0$ (not depending on the graph) such that, for any disk $D$ of radius $r$ that touches at most $d$ other disks, the other disks all have radius at least $\epsilon^dr$~\cite{MP94}.

We further recall the \emph{power diagram} for a set of disks in the plane \cite{A87}. The \emph{power distance} from a point $p$ to a disk $D$ with center $c$ and radius $r$ is $|p-c|^2 - r^2$. The power distance is negative if and only if $p$ is inside $D$, and otherwise it is the length of a tangent line from $D$ to $p$. Associated to $D$ is a convex polygon comprised of all points in the plane whose power distance to any disk in the configuration is minimized by its power distance to $D$. The power diagram is the set of these polygons. Since the disks in our configurations are disjoint, $D$ is contained in the interior of its polygon. Geometrically, the power diagram forms a subdivision of the plane into polygons, one per disk, with each disk interior to its polygon. It has $O(n)$ vertices and edges \cite[Lemma 1]{A87} and can be constructed in time $O(n \log n)$ \cite[\S5.1]{A87}.

\section{Monotonically Separated Disk Configurations}

We begin by fixing some notation and vocabulary. Suppose we have a
sequence of $n$ disks with centers $(x_i, y_i)$ for $1 \leq i \leq n$.
By rotating and relabeling the disks if necessary, we can assume
without loss of generality that $x_1 < x_2 < \cdots < x_n$, so the
disks are sorted by their $x$-coordinates.

\begin{definition}\label{def:monotonically-separated}
  We say that a sequence of disks is \emph{monotonically separated} if
  $x_1 < x_2 < \cdots < x_n$ and for every $i < j < k$, the $k$th disk
  is disjoint from the convex hull of the $i$th and $j$th disks, and
  the $i$th disk is disjoint from the convex hull of the $j$th and
  $k$th disks.
\end{definition}

\begin{figure}[ht]
\centering\includegraphics[scale=0.50]{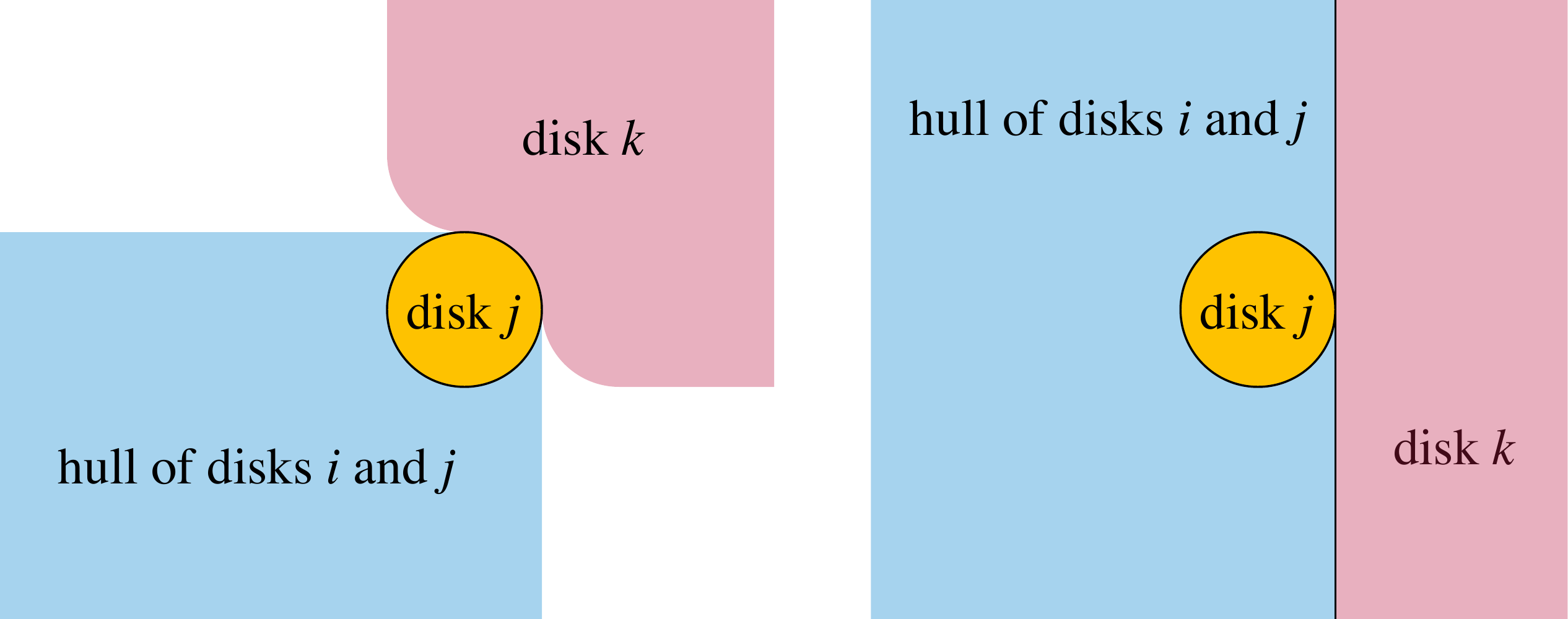}
\caption{Diagrams of monotonically separated sequences. For unit disks
  with indices $i<j<k$ and with centers that are $x$- and $y$-monotone
  (left) or whose $x$-coordinates differ by at least two units
  (right), the convex hull of disks $i$ and $j$ is confined to a
  disjoint region from disk $k$.}
\label{fig:monsep}
\end{figure}

\begin{lemma}
Given a sequence of $n$ unit disks with $x_1 < \cdots < x_n$, the disks are monotonically separated in either of the following two cases:
\begin{enumerate}
  \item ($xy$-monotone) $y_1 < \cdots < y_n$, so the centers are both $x$-monotone and $y$-monotone; or
  \item ($x$-separated) for all $1 \leq i < n$, we have $x_{i+1} - x_i \geq 2$.
\end{enumerate}
\end{lemma}

\begin{proof}
In the case where the $x$-coordinates differ by at least two units, the vertical line $x=x_j+1$ separates the convex hull of disks $i$ and $j$ from disk $k$. See \Cref{fig:monsep}. The other cases are similar.
\end{proof}

\begin{theorem}
\label{thm:separated-has-belt}
Every monotonically separated sequence of unit disks has a conveyor belt.
\end{theorem}

\begin{proof}
  Consider the \emph{upper convex hull} of the disks,
  which is the part of the boundary of the convex hull passing
  clockwise from the leftmost point of the leftmost disk $L$
  to the rightmost point of the rightmost disk $R$, not including
  these endpoints. Let $U_1, U_2, \ldots$ be the disks that contact
  the upper convex hull in between $L$ and $R$, referred to as the
  \emph{upper hull disks} and listed with increasing $x$-coordinate. Let
  $D_0 = L, D_1, D_2, \ldots, D_k = R$ be the disks that are not in the
  subsequence of upper hull disks, listed with increasing $x$-coordinate.
  We assume $k \geq 3$ since otherwise the boundary of the convex hull of
  the disks is a conveyor belt.

  We begin by describing a key subroutine and some of its
  properties. We will shortly use this subroutine to iteratively construct a
  conveyor belt on the given disks.
  Given two consecutive disks $D_i$ and $D_{i+1}$, the ``winding process''
  constructs two partial conveyor belts from $D_i$ to $D_{i+1}$ as follows.
  First, place a ray along the 
  ``lower'' bitangent between $D_i$ and $D_{i+1}$ with the ray's apex
  on $D_i$ and pointing the ray initially towards $D_{i+1}$.
  Continuously rotate the ray counterclockwise so that it stays tangent to
  $D_i$ at its apex. As we rotate, we may shift
  the apex of the ray to disks other than $D_i$ according to the following
  possible events considered in order.
  \begin{enumerate}
    \item The rotating ray may align with a tangent to some upper hull disk
     $U$ that lies between $D_i$ and $D_{i+1}$ in the $x$-sorted order of
     disks. If there is more than one such $U$, choose the one with the smallest
     $x$-coordinate. In this case, add to the current partial conveyor belt
     the bitangent along the ray between the disk the ray is currently rotating
     around and $U$. Shift the apex of the rotating ray to $U$
     and continue rotating the apex of the ray counterclockwise around $U$
     and inscribing an arc along the boundary of $U$.
   \item The rotating ray may align with a tangent to $D_{i+1}$ for the first time.
     In this case, we report the partial conveyor belt from $D_i$
     through the apex of the ray to this point of tangency as the ``lower'' of
     the two partial conveyor belts between $D_i$ and $D_{i+1}$.
     Continue rotating the ray.
   \item The rotating ray may align with a tangent to $D_{i+1}$ for the second time.
     In this case, report a partial conveyor belt from $D_i$ to $D_{i+1}$ as before, this
     time calling it the ``upper'' partial conveyor belt. Stop the winding process.
  \end{enumerate}

\bigskip

We claim the winding process has the following properties.
  \begin{enumerate}[(a)]
    \item \label{hull} Both partial conveyor belts lie entirely within the convex hull
      of $D_i$ and $D_{i+1}$. More precisely, they lie within the
      wedge-shaped region bounded by the two tangent lines to
      $D_{i+1}$ from the initial starting point on $D_i$.
    \item \label{upper} Both belts are disjoint from the upper convex hull, even though some
      upper hull disks may be touched by either belt.
    \item \label{curve} Both belts are valid partial conveyor belts touching
      $D_i$ and $D_{i+1}$ in the sense that
      they consist of arcs of disks and bitangents between them whose
      union is a continuously differentiable curve without self-intersection
      that is disjoint from all disk interiors.
  \end{enumerate}
  Property (\ref{hull}) is easy to see, and (\ref{upper}) follows from (\ref{hull}).
  For (\ref{curve}), every property is clear except possibly that the partial belts may
  not be disjoint from all disk interiors. The ray initially points along the
  ``lower'' bitangent between $D_i$ and $D_{i+1}$, which cannot intersect
  the interior of an upper hull disk by definition of the upper hull
  so it is reported as part of the lower belt.  As the winding process
  continues, the upper belt remains disjoint from
  all upper hull disk interiors with center $x$-coordinate between $D_i$ and $D_{i+1}$
  by construction. The remaining disks $D_j$ are disjoint from the convex hull of
  $D_i$ and $D_{i+1}$ since the disks are monotonically separated by hypothesis, so both
  belts are disjoint from these remaining disks by (\ref{hull}).
  
  The ``unwinding process'' is a slight variation on the
  winding process. Initially, the ray is placed exactly as before along
  the ``lower'' bitangent from $D_i$ to $D_{i+1}$ except
  its direction is reversed, with the apex still at $D_i$. The apex rotates counterclockwise,
  continuing exactly as before. The belts created by the unwinding
  process also satisfy (\ref{hull})-(\ref{curve}) above by virtually the same reasoning.
  
  Now consider any triple of disks $D_{i-1}, D_i, D_{i+1}$.
  The interior of the convex hull of $D_{i-1}$ and $D_i$ intersects the boundary
  of $D_i$ in an open semi-circle, and the interior of the convex hull of $D_i$ and $D_{i+1}$
  intersects the boundary of $D_i$ in another open semicircle. Suppose
  for simplicity that the centers of $D_{i-1}, D_i, D_{i+1}$ are
  not colinear. The union of these two semicircles is then a single arc on
  the boundary of $D_i$.
  Let $\alpha$ denote the complement of this open arc in the boundary of $D_i$.
  We may now paste
  together two of the four partial belts from $D_{i-1}$ to $D_i$,
  an arc containing $\alpha$, and two of the four partial belts
  from $D_i$ to $D_{i+1}$ as in \Cref{fig:glued_quads},
  resulting in four partial belts from $D_{i-1}$ to $D_{i+1}$.
  Since the disks are monotonically separated, one may check that property (\ref{hull})
  ensures each resulting partial belt from $D_{i-1}$ to $D_{i+1}$
  has no self-intersections and continues to satisfy (\ref{upper}) and (\ref{curve}).

\begin{figure}[ht]
\centering\includegraphics[width=0.49\textwidth]{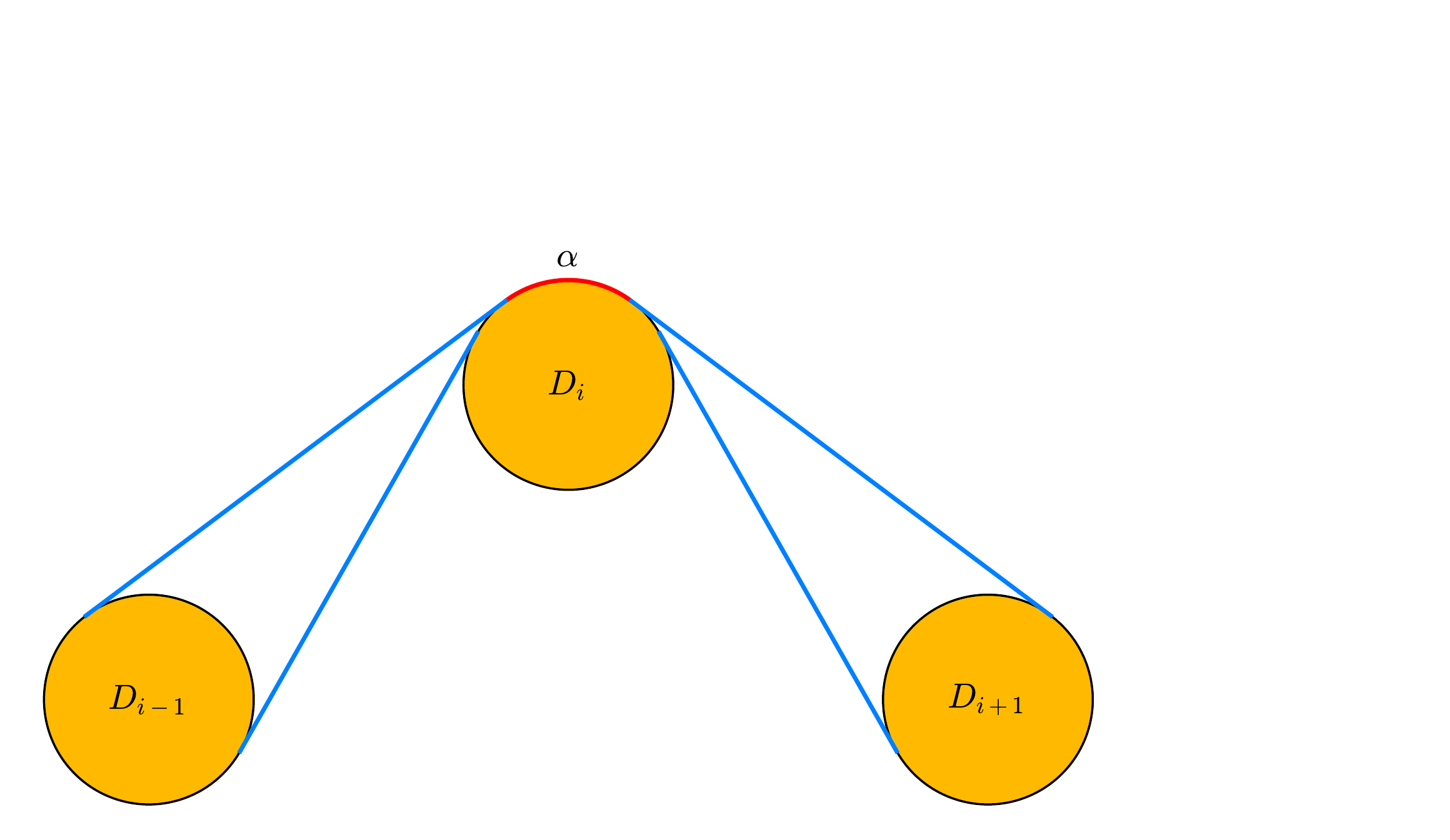}
\centering\includegraphics[width=0.49\textwidth]{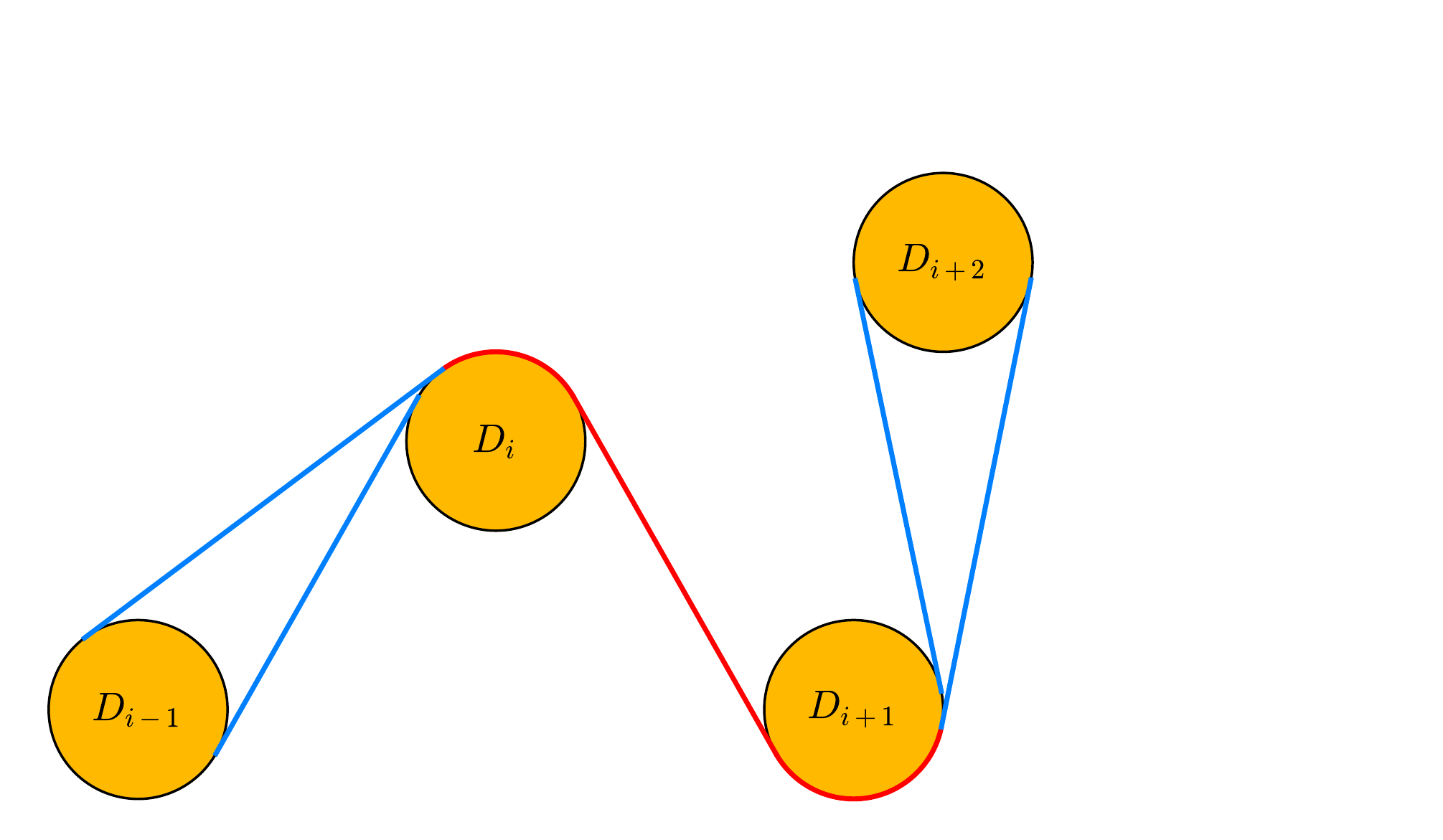}
\caption{The first picture shows the four partial conveyor belts
between three adjacent lower disks.  The red arc labeled $\alpha$ is
required to be part of the final belt.   The second picture shows the
next iteration of the algorithm after extending the 4 possible belts
to include $D_{i+2}$.  Again the red part of the belt is required in
any further growth of the belt, and the blue forks indicate the 4
optional extensions.}
\label{fig:glued_quads}
\end{figure}

  The full algorithm proceeds as follows. Begin with four partial
conveyor belts between two adjacent lower disks. Iteratively extend
these belts as above to produce four partial conveyor belts between
ever larger sets of adjacent lower disks satisfying (\ref{upper}) and
(\ref{curve}). For example, see \Cref{fig:glued_quads}.  Finally,
choose the belt tangent to the lower sides of $L$ and $R$ and glue it
to the upper hull by way of arcs along $L$ and $R$ to create a valid
conveyor belt touching each disk at least once.
\end{proof}

\begin{figure}[ht]
\centering\includegraphics[width=0.6\textwidth]{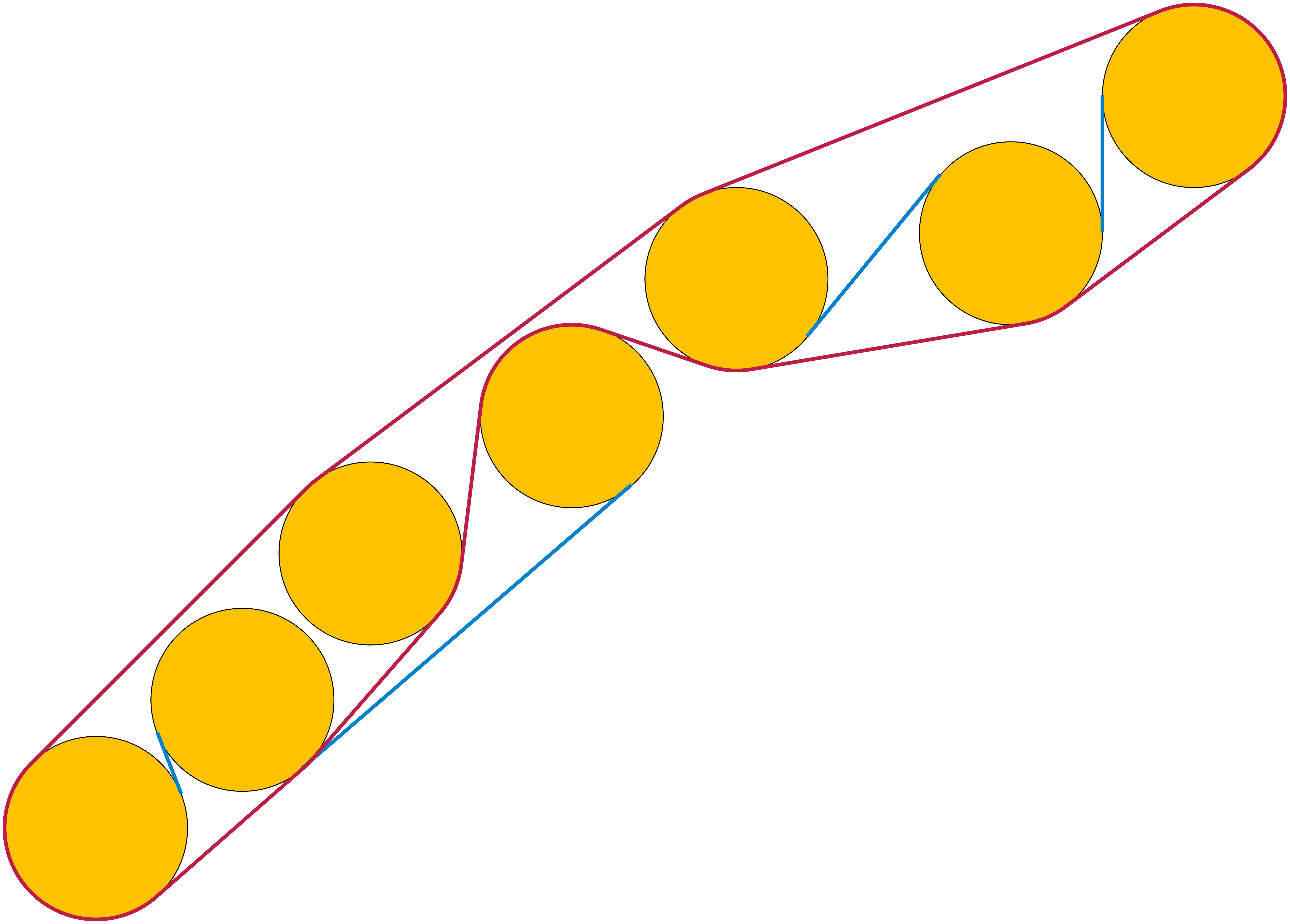}
\caption{A sequence of $x$- and $y$-monotone unit disks and the conveyor belt (red) produced for it by our algorithm. The blue line segments indicate the alternative bitangents found as preliminary results by our algorithm, connecting to the opposite side of each non-upper-hull disk from the bitangent that is eventually used in the conveyor belt.}
\label{fig:monotone}
\end{figure}

\Cref{fig:monotone} depicts the conveyor belt resulting from the
construction in the proof in red.  The blue segments are part of an
upper or lower belt that was unused in the final conveyor belt.
\Cref{fig:ex_monotone} shows another example of the final belt
constructed in the algorithm where the disks are monotonically
separated but not $x$- and $y$-monotone.

\begin{figure}[ht]
\centering\includegraphics[width=0.6\textwidth]{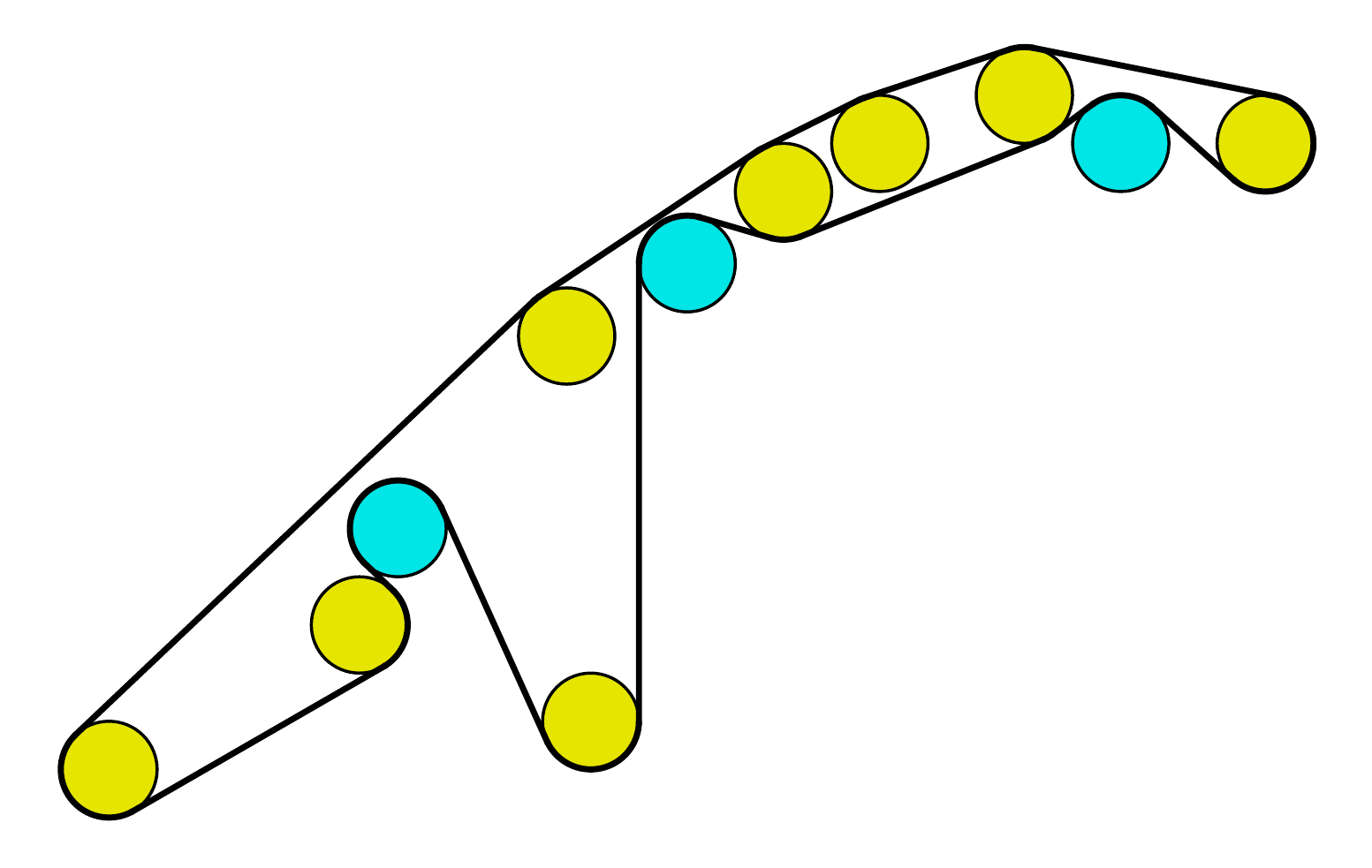}
\caption{Larger example of the conveyor belt produced by the proof of
\Cref{thm:separated-has-belt}.}
\label{fig:ex_monotone}
\end{figure}

\begin{theorem}\label{thm:separated-linear-runtime}
  The conveyor belt constructed by the algorithm in the proof of
  \Cref{thm:separated-has-belt} can be constructed in linear time
  from an input that lists the coordinates of the disk centers in the
  order of their monotonically separated sequence.
\end{theorem}

\begin{proof}
The upper convex hull can be constructed from this sorted order by a standard Graham scan algorithm, in Andrew's variation using the left-to-right sorted order rather than radial order~\cite{And-IPL-79}.
In this algorithm, we use a stack data structure to store the sequence of disks in the upper hull of ever-longer prefixes of the input. To extend the prefix, we check if the next disk in the sequence together with the top two disks on the stack form an upper hull that is unchanged if the middle of these three disks is omitted. While this is true we pop the stack, effectively removing the middle disk from consideration. After no more such pops can be performed, we push the new disk onto the stack.
At the end of this process, the disks that remain on the stack form the upper hull. Each disk is pushed and popped at most once, so the total time for this part of the algorithm is linear.

To simulate the continuous ray-rotation process by which we generated the partial conveyor belts through the remaining disks, we replace it by a discrete process consisting only of the events at which the continuous process undergoes a discrete change, when the rotating ray reaches a tangency with $D_{i+1}$ or with one of the disks $U_j$.  There may be many disks $U_j$ between $D_i$ and $D_{i+1}$; however, by convexity of the sequence of disks $U_j$, only two of these disks can cause tangency events: the one that is closest in the sequence to $D_i$, and the one that is closest in the sequence to $D_{i+1}$. Therefore, constructing the sequence of discrete events entails only comparing $O(1)$ possible tangencies at each step to determine which of them happens first, and can be performed in constant time per step. Each disk is visited at most once during this process, so the total time for this part of the algorithm is again linear.
\end{proof}

\begin{figure}[t]
\centering\includegraphics[width=0.55\textwidth]{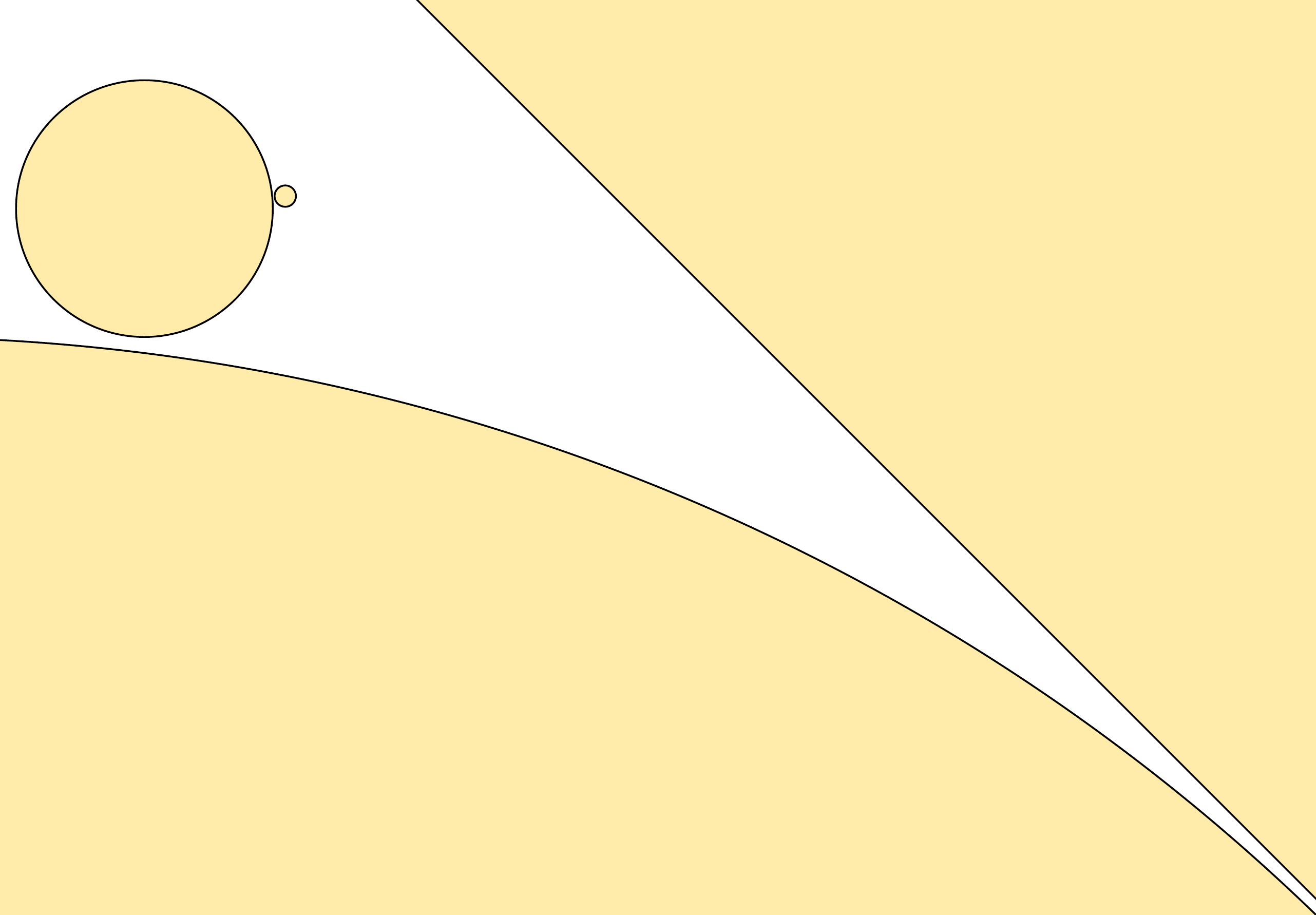}\quad\qquad
\includegraphics[width=0.1\textwidth]{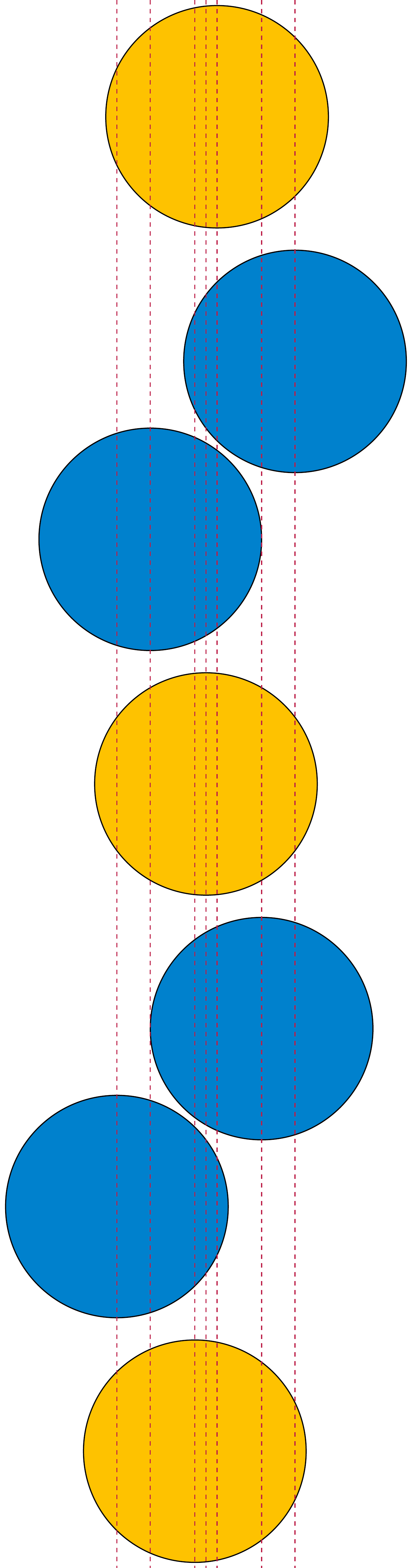}
\caption{Left: A sequence of $x$- and $y$-monotone disks with no bitonic conveyor belt. The second-to-last disk in this sequence sorted by $x$-coordinates of centers (the smallest disk) has no bitangents to the last and largest disk, so it must be a local but not global maximum in the sequence of indices of any conveyor belt.
Right: A sequence of unit disks with no bitonic conveyor belt. The three yellow disks have no unblocked bitangents, and must be separated from each other by at least three local minima or maxima.}
\label{fig:monotone-blocked}
\end{figure}

A finite sequence of distinct numbers is called \emph{bitonic} if it has a unique local maximum and a unique local minimum when viewed as a cyclic sequence. Equivalently, it is \emph{unimodal} up to a cyclic rotation. Since our method constructs conveyor belts from left-to-right, the result is bitonic with respect to the sequence of $x$-coordinates of disk centers. When a set of disks has a bitonic conveyor belt, it is possible to find such a belt in polynomial time by a straightforward adaptation of the well-known dynamic programming algorithm for bitonic traveling salesperson tours~\cite{CLRS}. However, for disks with $x$- and $y$-monotone centers but non-uniform radii, it is not always possible to find conveyor belts that are bitonic. \Cref{fig:monotone-blocked} (left) provides a counterexample. For unit disks that are not monotonically separated, it is also not generally possible to have a tour that is bitonic. In \Cref{fig:monotone-blocked} (right) there are three disks that are consecutive in the sorted order by $x$-coordinates, but that have no unblocked bitangent between any two of them. All three disks must be separated from each other in the cyclic sequence of contacts with the conveyor belt, and between any two of them there must be a local minimum or local maximum of the $x$-coordinates.

\section{One-Touch NP-Completeness}

In this section we consider the following \emph{one-touch conveyor belt problem}.
We are given as input a collection of disjoint disks in the plane specified by integer center coordinates and radii. The goal is to determine whether there exists a conveyor belt that contacts each disk exactly once.

Our proof that this problem is NP-complete will serve as a warm up for the proof that the decision problem for conveyor belts that are not restricted to one touch is also NP-complete.

\begin{theorem}
\label{thm:one-touch-completeness}
The one-touch conveyor belt problem is NP-complete.
\end{theorem}

\begin{proof}
We follow the standard outline for a proof that a problem $X$ is
NP-complete, by a polynomial-time reduction from a known NP-complete
problem $Y$. To do so $X$ and $Y$ must both be decision problems
(problems with a yes or no answer). We must show that $X$ is in NP,
and that there exists a polynomial time algorithm for translating all inputs of $Y$ into inputs of $X$ that preserves the answer of each translated input.

To show that the one-touch conveyor belt problem is in NP, we describe a nondeterministic polynomial-time algorithm for it: an algorithm that guesses a \emph{solution} that can be described in a polynomial number of bits, verifies in deterministic polynomial time that the solution is valid, and if so answers yes. It must be the case that every solvable instance has a solution that causes this algorithm to answer yes, and that no non-solvable instance has such a solution. In our case, the solutions can specify the cyclic order of contacts along the belts and the orientations of each disk, information that can be specified in only $O(n\log n)$ bits for $n$ disks, a polynomial number. If there are $n$ disks, since the belt is one-touch, the resulting belt will have $2n$ arcs or bitangent segments. The verification algorithm must then test that only consecutive arcs or bitangents intersect, and then only in a single point with compatible directions. It is straightforward to see that these verifications may be performed with $O(n^2)$ tests, so the verification algorithm takes polynomial time, as required.

To prove NP-hardness, we reduce from a known NP-complete problem, determining the existence of a Hamiltonian cycle in a maximal planar graph. This problem was proven NP-complete by Wigderson in 1982~\cite{W82}. It is straightforward to modify Wigderson's construction to ensure that the maximal planar graphs that it produces always have an even number of vertices; see \Cref{rem:Wigderson}.

The reduction begins by representing a given maximal planar graph as a
system of touching disks using the circle packing theorem. The only
unblocked bitangent segments in the circle packing are outer tangents
for the three outermost disks. Shrinking each disk by a factor of
$1-\delta$ for $\delta$ sufficiently small will create additional unblocked
bitangents only between adjacent disks, see \Cref{fig:octahedral-pulley}.
When finding the largest possible $\delta$ for a given triple of
disks, we will perform a constant number of polynomial operations and
root extractions using the disk radii, so $\delta$ is at least a
polynomial in the smallest possible ratio $\epsilon^n$ between tangent
disks.


\begin{figure}[t]
\centering\includegraphics[width=0.65\textwidth]{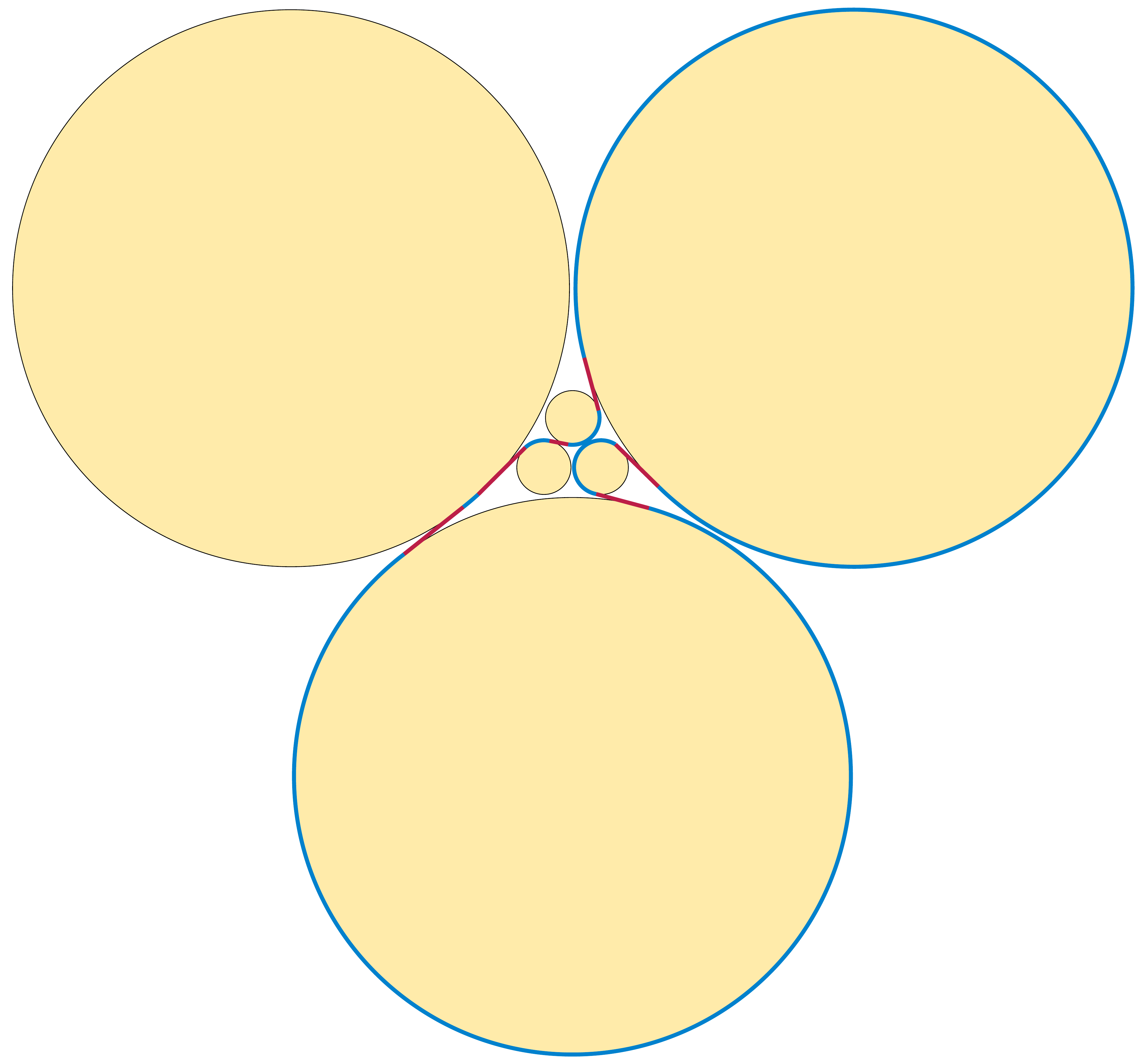}
\caption{A one-touch conveyor belt for a system of disks constructed by slightly shrinking a circle packing for the graph of the edges and vertices of an octahedron. The belt touches the disks in the order given by a Hamiltonian cycle on the octahedron, and crosses between disks using only inner tangents. Circular arcs on the belt are shown in blue, and bitangents are shown in red.}
\label{fig:octahedral-pulley}
\end{figure}

The numerical precision needed to represent the disk centers and their radii accurately enough to perform this shrinkage, and to avoid creating additional bitangencies through numerical inaccuracies, is therefore also at most exponential in the number of disks. Numbers with this level of precision may be represented using a linear number of bits, allowing an approximate numerical representation of the circle packing and its shrunken system of disjoint disks to be constructed numerically in polynomial time. By scaling this numerical representation appropriately, we can cause all the disk centers and radii to become integers, as needed for an input to the one-touch conveyor belt problem.

It remains to verify that this transformation from graphs to systems
of disks preserves the yes-or-no answers to every input. That is, we
must show that a given maximal planar graph $G$ with an even number of
vertices has a Hamiltonian cycle if and only if the resulting system
of disks has a one-touch conveyor belt. We begin by arguing that, when
a one-touch conveyor belt exists, a Hamiltonian cycle also
exists. Thus, suppose that we have a one-touch conveyor belt
$B$. Since every face of $G$ is a triangle, the only unobstructed
bitangents that exist for the disks of the construction are between
shrunken disks that, before shrinking, were tangent. Therefore, $B$
can only use these segments to move from one disk to another, and each
two consecutive disks along $B$ must be adjacent in $G$.  Thus, the
cyclic ordering of disks in $B$ corresponds to a cyclic ordering of
vertices in $G$ such that each two consecutive vertices in the
ordering are adjacent; that is, to a Hamiltonian cycle.

Conversely, suppose we have a Hamiltonian cycle $C$ in $G$; we must show that there exists a one-touch conveyor belt in the corresponding system of shrunken disks. To do so, we use the cyclic ordering of vertices in $C$ as the cyclic order of disks in a belt $B$, and we assign the disks signs that alternate between positive and negative in this cyclic order. This assignment is well-defined since the number of vertices in $G$ is assumed to be even. It corresponds to a curve composed of circular arcs and inner bitangents of pairs of circles that, before shrinking, were tangent. The shrinking process that we perform causes all such inner bitangents to be unobstructed, and disjoint from all the other inner bitangents of other formerly-tangent pairs, so no two such segments can cross, nor can a bitangent segment cross any circular arc. Therefore, the resulting curve is a simple closed curve, composed of arcs and bitangents, with one arc per disk, so it is a one-touch conveyor belt as required.
\end{proof}

\begin{remark}\label{rem:Wigderson}
  We may modify Wigderson's construction~\cite{W82} to produce graphs with evenly many nodes as follows. From Wigderson's graph $K$, one may create a graph $K'$ by adding an extra vertex at the midpoint of the edge of $K$ that starts at the central vertex and travels northwest, followed by triangulating the two resulting rectangles by adding two extra edges emanating from the new vertex. In forming Wigderson's graph $M$, use one copy of $K$ and one copy of $K'$. The resulting variant of Wigderson's graph $N$ will have $58$ instead of $55$ nodes, so the graphs $G'$ will have a multiple of $58$ nodes, an even number.
\end{remark}

\Cref{fig:octahedral-pulley} depicts this construction for the maximal planar graph given by the vertices and edges of a regular octahedron. The one-touch belt shown in the figure connects the six disks shown in the order corresponding to a Hamiltonian cycle through the six vertices of the octahedron.

\begin{remark}\label{rem:counting.problem}
Whether counting the number of solutions is \#P-complete remains open. The NP-completeness reduction above is close to parsimonious, but it counts cycles through edges of the outer triangle differently than cycles that do not use those edges. More precisely, one may check that the number of conveyor belts corresponding to a given Hamiltonian cycle is $1$ if the cycle involves $0$ or $1$ edges of the outer triangle, and $2$ if the cycle involves $2$ edges of the outer triangle.
\end{remark}

\section{Multi-Touch NP-Completeness}

\begin{theorem}
\label{thm:multi-touch-completeness}
It is NP-complete to determine whether a given system of disks has a conveyor belt, even allowing the belt to touch a single disk multiple times along disjoint arcs.
\end{theorem}

\begin{figure}[t]
\centering\includegraphics[scale=0.35]{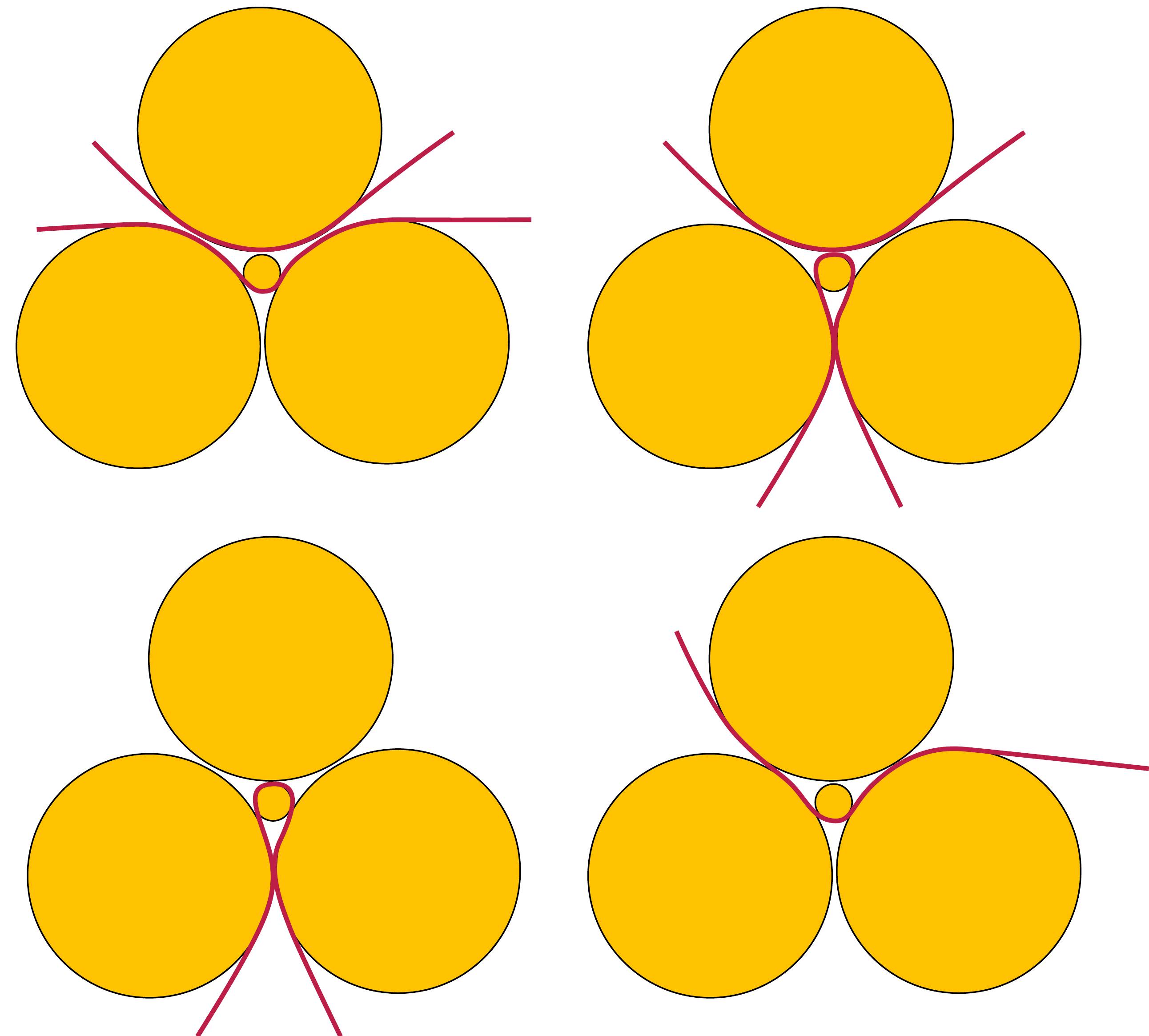}
\caption{The four possible patterns for a conveyor belt to enter and exit a triangle of a circle packing with a fourth smaller circle inside it: Two double-ply crossings to adjacent triangles (top left), one double-ply crossing (bottom left), a double-ply and two single-ply crossings (top right), or two single-ply crossings (bottom right). The single-ply crossings can be tangent to either of the two circles between which they pass. It is not possible for part of the belt to enter on a double-ply crossing and exit on a single-ply crossing, or vice versa.}
\label{fig:pulley-gadget-paths}
\end{figure}

\begin{figure}[t]
\centering\includegraphics[scale=0.28]{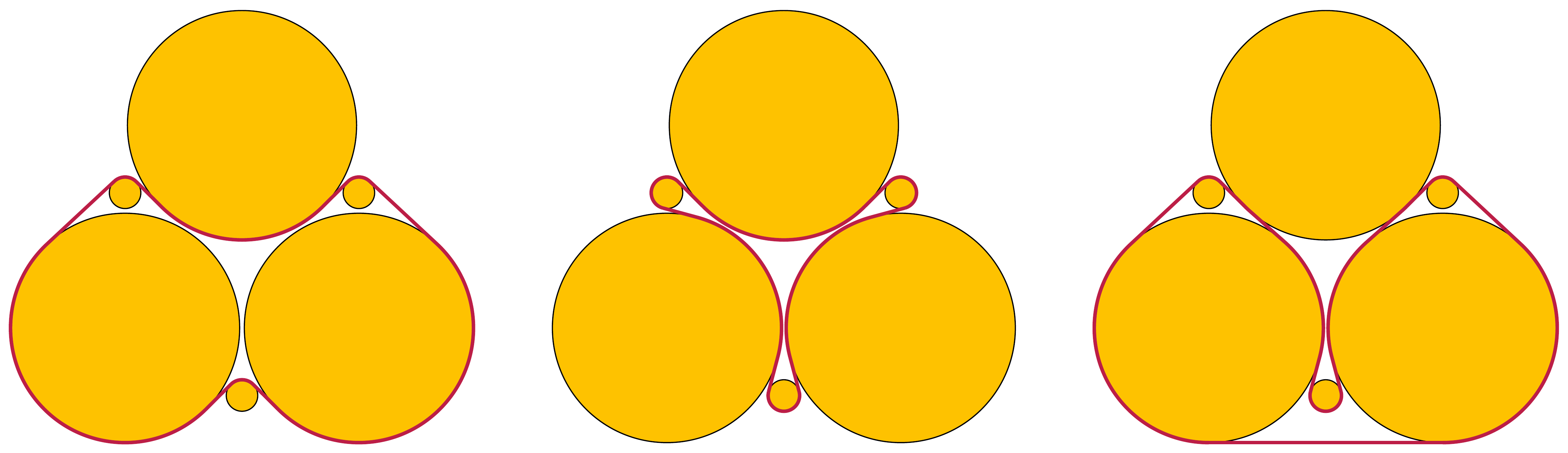}
\caption{Adding three small disks to the outer triangle of a circle packing restricts a conveyor belt to enter and exit the interior of the packing as shown, up to rotation and reflection.}
\label{fig:outside-triangle-gadget}
\end{figure}

\begin{proof}
As in the one-touch case, we reduce from the Hamiltonian circuit problem, but in a different class of graphs, the cubic 3-connected planar graphs. These graphs are the dual graphs of the maximal planar graphs used for the one-touch problem. Again, testing Hamiltonicity of these graphs is known to be NP-complete~\cite{GarJohTar-SICOMP-76}; indeed, Wigderson's reduction in \cite{W82} is to this older result.

Our reduction is similar to the one-touch reduction. Let $G$ be a
cubic 3-connected planar graph with $n>3$ vertices.  Let $G'$ be the
maximal planar graph dual to $G$.  Represent $G'$ by a circle packing,
and then slightly shrink the disks. This time, we modify the
construction by adding another small disk at the radical center of
each triple of disks that represent an interior triangle of $G'$
 (\Cref{fig:pulley-gadget-paths}), and by adding
three more disks tangent to pairs of the outer three disks, small
enough that they lie within the convex hull of the outer three disks
(\Cref{fig:outside-triangle-gadget}). The intent of these changes is
to force the conveyor belt to visit every face of $G'$ and thereby
represent a Hamiltonian cycle for the originally given cubic graph $G$,
whose vertices correspond to these faces. Since $G$ has $n>3$ vertices,
there is at least one interior triple of disks.

For each pair of disks that were initially tangent, a conveyor belt may pass between them at most twice. Call such a crossing \emph{single-ply} if the conveyor belt passes between those disks once and \emph{double-ply} if the conveyor belt passes between those disks twice. All of the ways that a belt can enter and exit one of the inner triangles of the circle packing are depicted in \Cref{fig:pulley-gadget-paths}. By inspecting each case, we find that a belt enters through a single-ply crossing if and only if it exits through a single-ply crossing. It follows that any segment of the conveyor belt passing through the interior of the outer triangle of the circle packing must consist entirely of single-ply or entirely of double-ply crossings.

The disks added for the outside triangle can only be touched by a conveyor belt as shown in \Cref{fig:outside-triangle-gadget}. We consider each depicted case in turn.

\begin{itemize}
\item In \Cref{fig:outside-triangle-gadget} (left), the conveyor belt
enters and exits the interior of the disk construction via two
single-ply connections in regions which correspond to two
neighboring triangles of the outer triangular face of
$G'$. Consequently, if a one-touch conveyor belt exists for the disk
construction, it must be single-ply everywhere in the interior as
shown in the unique single-ply connection pattern from
\Cref{fig:pulley-gadget-paths} (lower right).  Furthermore, the
conveyor belt must enter and exit each region corresponding with a
triangle in $G'$ exactly once, so it must correspond to a
Hamiltonian circuit of $G$. Conversely, it is not hard to see that
every Hamiltonian circuit of $G$ can be transformed to a conveyor belt
on the constructed disk configuration. Therefore, the disk
configuration derived from $G$ has a conveyor belt if and only if $G$
has a Hamiltonian circuit.

\item In \Cref{fig:outside-triangle-gadget} (middle), the conveyor
belt enters and exits the interior of the constructed disk
configuration via three double-ply connections. The interior
connections must then all be double-ply as well, as depicted in
\Cref{fig:pulley-gadget-paths} (upper left) and (lower
left). Furthermore, we can observe from the pictures that both strands
enter together from one gap between disks and exit together from another gap between
disks.  In this sense, double-ply segments form paths.  Starting at one
of the double-ply connections entering the interior from the outside
triangle, the belt must form a path through a sequence of regions
corresponding with triangles in $G'$. If the path ended at
\Cref{fig:pulley-gadget-paths} (lower left), there would be a
disconnected segment of the belt, so it must end at another connection
to the outside triangle. However, the third connection to the outside
triangle would then be disconnected from the other two. So, there are
no belts of this form.

\item In \Cref{fig:outside-triangle-gadget} (right), the conveyor belt
makes two single-ply connections and one double-ply connection to the
interior of the construction. The double-ply connection must pass
through the interior to the single-ply connections in order for the
belt to be connected, though as noted above the interior connections
are all single-ply or all double-ply, so this is not possible and
there are again no belts of this form.
\end{itemize}

Thus, the only possible conveyor belts are ones that makes single-ply
connections only, each of which correspond to a Hamiltonian cycle.
\end{proof}

\section{Guide Disks}

In this section, we explore the question of how many additional disks
must be added in order to guarantee that any disk configuration has a
conveyor belt.  We show that a linear number of additional disks
suffice.  Our construction uses the power diagram of a disk
configuration described in \Cref{sec:preliminaries}.

\begin{theorem}
\label{thm:guide-disks-suffice}
Any system of $n$ disks can be augmented by $O(n)$ additional disks so that the resulting augmented system of disks has a (one-touch) conveyor belt.
\end{theorem}

\begin{proof}
First construct the power diagram of the disks. The dual graph of the
power diagram is the graph whose nodes are polygons in the power
diagram, where nodes are connected if their polygons share an
edge. Pick a spanning tree of this dual graph. Imagine traveling
around the ``outside'' of the spanning tree, which gives a cyclic
sequence $S$ of disks and edges of the dual graph. The disks in $S$
may be repeated, consecutive disks have power diagram cells that are
adjacent to each other, and each disk is included at least once in
$S$. Since the number of edges of the power diagram is linear in $n$,
the length of $S$ is linear in $n$.

Let $C$ and $D$ be consecutive disks in $S$ corresponding to an edge
in the spanning tree. Take $y$ to be the
midpoint of the edge between the polygons of $C$ and $D$ in the power
diagram; if the edge is infinite, we may take $y$ to be any point on
the interior of the edge. Let $x$ and $z$ be points on the line
segments from the centers of $C$ and $D$ to $y$ just outside of $C$
and $D$, respectively. We may represent the edge between $C$ and $D$
in the tree geometrically by a polyline consisting of the two line segments
$xy$ and $yz$. The polylines representing distinct edges do not
cross. Taken together, they geometrically represent the spanning tree
of the dual graph. By adding small guide disks near $x$, $y$, and $z$
as needed, we may form a conveyor belt that represents traveling
around the ``outside'' of the spanning tree. Since the length of $S$
is linear, $O(n)$ guide disks are needed.

In the one-touch model, we may perform a similar construction, but when it would use more than one arc of an input disk we may instead route the belt near the edges of the corresponding power diagram polygon for all but one of these arcs. This will require at most a constant number of disks per edge in the power diagram, so there are again $O(n)$ guide disks total.
\end{proof}

\begin{figure}[t]
\centering\includegraphics[width=0.5\textwidth]{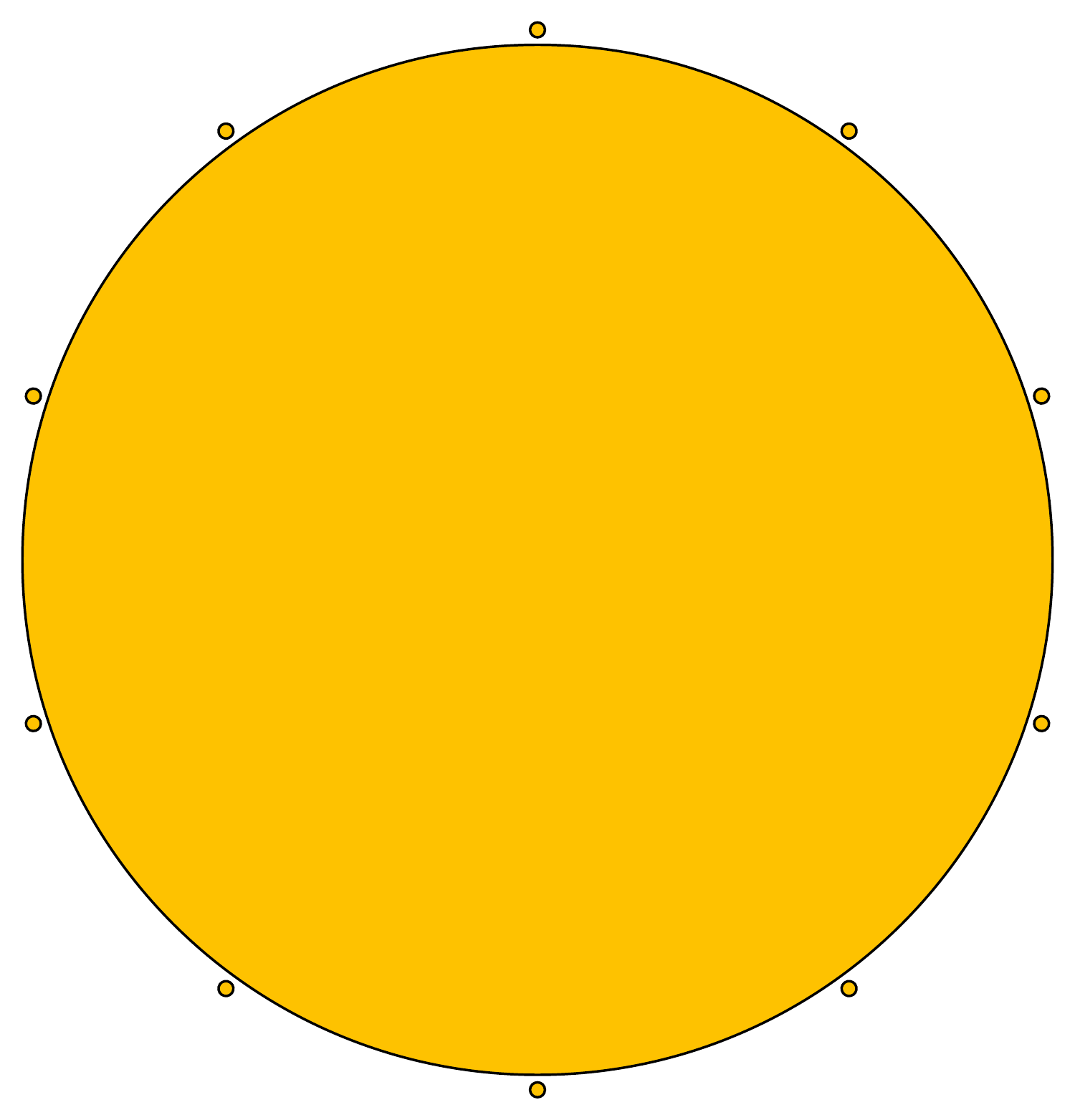}
\caption{A system of $n$ disks requiring $\Omega(n)$ guide disks for a one-touch conveyor belt}
\label{fig:guide-lb}
\end{figure}

\begin{lemma}
\label{lem:one-guide}
There exists a configuration $C$ of disks with the property that,
whenever $C$ appears as part of a larger configuration $C'$, with all disks of $C'\setminus C$ disjoint from the convex hull of $C$, every conveyor belt for $C'$ must include at least one guide disk interior to the convex hull of $C$.
\end{lemma}

\begin{proof}
Let $G$ be a cubic 3-connected planar graph in which the longest path has fewer than $|V(G)|/3$ vertices; the existence of $G$ follows from known results on the shortness exponent of cubic 3-connected planar graphs \cite{GruWal-JCTA-73}. As in the proof of \Cref{thm:multi-touch-completeness}, represent the dual graph of $G$ by a circle packing (choosing arbitrarily which triple of circles to make exterior), place another smaller tangent circle within each triangle of circles of the packing, and then shrink all of the circles by a small amount to allow a conveyor belt to pass between them without making any additional bitangents available to the belt. Let $C$ be the resulting system of circles. 

Then, as in the proof of \Cref{thm:multi-touch-completeness}, any conveyor belt must enter $C$ via at most three single-ply or double-ply crossings between the three pairs of outer circles of $C$.
If there were no guide disks within the convex hull of $C$, the same analysis in \Cref{thm:multi-touch-completeness} would show that such a belt would necessarily correspond to at most three paths through $G$ that either enter $C$ through one of these crossings and exit through another, or enter $C$ through a double-ply crossing and stop somewhere within $C$. However, because all paths are short, it is not possible for three or fewer paths to touch all vertices of $G$; correspondingly, a conveyor belt without a guide disk within $C$ cannot touch all of the smaller circles within each triangle of circles of $C$. Therefore, at least one guide disk within the convex hull of $C$ is needed.
\end{proof}

\begin{theorem}
\label{thm:guide-disks-necessary}
In order for one-touch or unconstrained conveyor belts to include $n$ given disks, $\Omega(n)$ guide disks are sometimes necessary.
\end{theorem}

\begin{proof}
For the one-touch case,
place $n-1$ small disks at the vertices of a regular polygon, and a large disk at the center of the polygon, in such a way that no two small disks can see each other (\Cref{fig:guide-lb}).  Then, in the cyclic sequence of disks given by any one-touch conveyor belt, each of the small disks must be separated from each other by a contact with some other disk, but only one of those contacts can be the large disk. Therefore, there must be at least $n-2$ other guide disks to provide these contacts.

For the unconstrained case, we form a configuration of $\lfloor n/|C|\rfloor$ copies of the configuration $C$ described in \Cref{lem:one-guide}, with disjoint convex hulls. Each copy of $C$ requires a guide disk interior to its convex hull, so all the copies together require at least $\lfloor n/|C|\rfloor$ guide disks.
\end{proof}

\section{Future Work}\label{}
We conclude with a few  questions for further exploration:

\begin{enumerate}
\item How many guide disks are necessary to ensure that any traveling
salesman tour on the centers of $n$ disks can be a subsequence of a
conveyor belt?
\item Is the problem of finding the number of conveyor belts for a
given disk configuration \#P-complete?
\item Given $n$ points in the plane, is finding the number of polygons
with vertices at the $n$ points also \#P-complete?  See
\cite{Eppstein.2019} for related results and further questions on
hardness of counting problem in discrete geometry.
\end{enumerate}   

\section*{Acknowledgments}

This work was initiated during an open problem session held at the University of
Washington in 2016 while the Demaines were Walker--Ames Lecturers,
sponsored by the UW Graduate School.
We thank Paul Beame, Timea Tihanyi, Ron Irving, and Jamie Walker
for their support during the Walker-Ames nomination process;
and the 25 students and faculty in math, art, and computer science
who attended the session and tackled the problem.
This work was continued at the 32nd Bellairs Winter Workshop on Computational
Geometry, held in 2017 at the Bellairs Research Institute, Barbados.
We thank the other participants of that workshop for their input and
for creating an exciting atmosphere.

\bibliographystyle{abbrvurl}
\bibliography{belts}

\begin{thebibliography}{10}

\bibitem{Abe-GRSME-08}
M.~Abellanas.
\newblock {Conectando puntos: poligonizaciones y otros problemas relacionados}.
\newblock {\em Gaceta de la Real Sociedad Matematica Espa{\~n}ola},
  11(3):543{--}558, 2008.

\bibitem{Abe-EGC-11}
M.~Abellanas.
\newblock {Linking geometric objects}.
\newblock In P.~Ramos and V.~Sacrist{\'a}n, editors, {\em XIV Spanish Meeting
  on Computational Geometry, In Honor of Ferran Hurtado's 60th Birthday,
  Alcal{\'a} de Henares, June 27{--}30, 2011}, pages 31{--}32. Centre de
  Recerca Matem{\'a}tica, Bellaterra, 2011.
\newblock URL:
  \url{https://ddd.uab.cat/pub/llibres/2011/hdl_2072_200199/Documents08_web.pdf}.

\bibitem{And-IPL-79}
A.~M. Andrew.
\newblock {Another efficient algorithm for convex hulls in two dimensions}.
\newblock {\em Information Processing Letters}, 9(5):216{--}219, 1979.
\newblock \href {http://dx.doi.org/10.1016/0020-0190(79)90072-3}
  {\path{doi:10.1016/0020-0190(79)90072-3}}.

\bibitem{A87}
F.~Aurenhammer.
\newblock {Power diagrams: properties, algorithms and applications}.
\newblock {\em SIAM J. Comput.}, 16(1):78{--}96, 1987.
\newblock \href {http://dx.doi.org/10.1137/0216006}
  {\path{doi:10.1137/0216006}}.

\bibitem{BanDevEpp-JGAA-15}
M.~J. Bannister, W.~E. Devanny, D.~Eppstein, and M.~T. Goodrich.
\newblock {The Galois complexity of graph drawing: why numerical solutions are
  ubiquitous for force-directed, spectral, and circle packing drawings}.
\newblock {\em J. Graph Algorithms {\&} Applications}, 19(2):619{--}656, 2015.
\newblock \href {http://arxiv.org/abs/1408.1422} {\path{arXiv:1408.1422}},
  \href {http://dx.doi.org/10.7155/jgaa.00349} {\path{doi:10.7155/jgaa.00349}}.

\bibitem{CS03}
C.~R. Collins and K.~Stephenson.
\newblock {A circle packing algorithm}.
\newblock {\em Comp. Geom.}, 25(3):233{--}256, 2003.
\newblock \href {http://dx.doi.org/10.1016/S0925-7721(02)00099-8}
  {\path{doi:10.1016/S0925-7721(02)00099-8}}.

\bibitem{CLRS}
T.~H. Cormen, C.~E. Leiserson, R.~Rivest, and C.~Stein.
\newblock {\em {Introduction to Algorithms}}.
\newblock MIT Press, 3rd edition, 2009.

\bibitem{DemDem-TCS-15}
E.~D. Demaine and M.~L. Demaine.
\newblock {Fun with fonts: algorithmic typography}.
\newblock {\em Theoret. Comput. Sci.}, 586:111{--}119, 2015.
\newblock \href {http://dx.doi.org/10.1016/j.tcs.2015.01.054}
  {\path{doi:10.1016/j.tcs.2015.01.054}}.

\bibitem{DemDemPal-FE-10}
E.~D. Demaine, M.~L. Demaine, and B.~Palop.
\newblock {Conveyer-belt alphabet}.
\newblock In H.~Aardse and A.~van Baalen, editors, {\em Findings in
  Elasticity}, pages 86{--}89. Pars Foundation, Lars M{\"u}ller Publishers,
  April 2010.

\bibitem{Eppstein.2019}
D.~Eppstein.
\newblock Counting polygon triangulations is hard.
\newblock {\em CoRR}, abs/1903.04737, 2019.
\newblock Accepted to SoCG, June 2019.
\newblock URL: \url{http://arxiv.org/abs/1903.04737}.

\bibitem{GarJohTar-SICOMP-76}
M.~R. Garey, D.~S. Johnson, and R.~E. Tarjan.
\newblock {The planar Hamiltonian circuit problem is NP-complete}.
\newblock {\em SIAM J. Comput.}, 5(4):704{--}714, 1976.
\newblock \href {http://dx.doi.org/10.1137/0205049}
  {\path{doi:10.1137/0205049}}.

\bibitem{GruWal-JCTA-73}
B.~Gr{\"u}nbaum and H.~Walther.
\newblock {Shortness exponents of families of graphs}.
\newblock {\em Journal of Combinatorial Theory. Series A}, 14:364{--}385, 1973.
\newblock \href {http://dx.doi.org/10.1016/0097-3165(73)90012-5}
  {\path{doi:10.1016/0097-3165(73)90012-5}}.

\bibitem{MP94}
S.~Malitz and A.~Papakostas.
\newblock {On the angular resolution of planar graphs}.
\newblock {\em SIAM J. Discrete Math.}, 7(2):172{--}183, 1994.
\newblock \href {http://dx.doi.org/10.1137/S0895480193242931}
  {\path{doi:10.1137/S0895480193242931}}.

\bibitem{M93}
B.~Mohar.
\newblock {A polynomial time circle packing algorithm}.
\newblock {\em Discrete Math.}, 117(1{--}3):257{--}263, 1993.
\newblock \href {http://dx.doi.org/10.1016/0012-365X(93)90340-Y}
  {\path{doi:10.1016/0012-365X(93)90340-Y}}.

\bibitem{ORo-EGC-11}
J.~O'Rourke.
\newblock {String-wrapped rotating disks}.
\newblock In A.~M{\'a}rquez, P.~Ramos, and J.~Urrutia, editors, {\em XIV
  Spanish Meeting on Computational Geometry, EGC 2011, Dedicated to Ferran
  Hurtado on the Occasion of His 60th Birthday, Alcal{\'a} de Henares, Spain,
  June 27-30, 2011, Revised Selected Papers}, volume 7579 of {\em Lecture Notes
  in Computer Science}, pages 65{--}78, Cham, 2011. Springer.
\newblock \href {http://dx.doi.org/10.1007/978-3-642-34191-5_6}
  {\path{doi:10.1007/978-3-642-34191-5_6}}.

\bibitem{W82}
A.~Wigderson.
\newblock {The complexity of the Hamiltonian circuit problem for maximal planar
  graphs}.
\newblock Technical Report 298, Princeton University Department of Computer
  Science, 1982.

\bibitem{Ziegler-1995-Steinitz}
G.~M. Ziegler.
\newblock Steinitz' theorem for 3-polytopes.
\newblock In {\em Lectures on Polytopes}, {L}ecture~4, pages 103--126.
  Springer-Verlag, 1995.

\end{thebibliography}
\end{document}